\documentclass{article}
\usepackage{graphicx}
\usepackage{amsmath}
\usepackage{amsthm}
\usepackage{amssymb}

\newtheorem{theorem}{Theorem}[section]
\newtheorem{corollary}{Corollary}[theorem]
\newtheorem{proposition}{Proposition}[theorem]
\newtheorem{lemma}[theorem]{Lemma}
\newtheorem{remark}{Remark}

\usepackage{soul}
\usepackage{float}
\usepackage{verbatim}

\newcommand{\sumNCG}{\textsc{sum ncg}\xspace}

\newcommand{\NE}{\textsc{ne}\xspace}

\newcommand{\PoA}{\text{PoA}\xspace}

\newcommand{\remove}[1]{}

 \usepackage[margin=1in]{geometry}

\usepackage{xspace}
\usepackage{color}
\usepackage[utf8]{inputenc}
\usepackage{psfrag}

\usepackage[toc,page]{appendix}

\usepackage[normalem]{ulem}
\usepackage{soul}

\usepackage{array}
\usepackage{booktabs}
\usepackage{multirow}
\usepackage{hhline}
\usepackage{graphicx}
\usepackage{abraces}

\newtheorem*{theorem*}{Theorem}

\usepackage{authblk}

\title{New Insights into the Structure of Equilibria for the Network Creation Game}
\author[1]{\`Alvarez, C. \thanks{alvarez@cs.upc.edu}}
\affil[1]{ALBCOM Research Group, Computer Science Department, UPC, Barcelona}
\author[2]{Messegu\'e, A. \thanks{amessegue@cs.upc.edu}}%
\affil[2]{ALBCOM Research Group, Computer Science Department, UPC, Barcelona}
\date{}                     
\begin{document}

\maketitle              

\begin{abstract}

We study the \emph{sum classic network creation game} introduced by Fabrikant et al. in which $n$ players conform a network buying links at individual price $\alpha$. When studying this model we are mostly interested in \emph{Nash equilibria} (\NE) and the \emph{Price of Anarchy} (\PoA). It is conjectured that the \PoA is constant for any $\alpha$.  Up until now, it has been proved constant \PoA for the range  $\alpha = O(n^{1-\delta_1})$ with $\delta_1>0$ a positive constant, upper bounding by a constant the diameter of any \NE graph jointly with the fact that the diameter of any \NE graph plus one unit is an upper bound for the $\PoA$ of the same graph. Also, it has been proved constant \PoA for the range $\alpha >n(1+\delta_2)$ with $\delta_2>0$ a positive constant, studying extensively the average degree of any biconnected component from equilibria.

Our contribution consists in proving that \NE graphs satisfy very restrictive topological properties generalising some properties proved in the literature and providing new insights that might help settling the conjecture that the \PoA is constant for the remaining range of $\alpha$:

(i) We show that \emph{every} node has the majority of the other nodes of the \NE graph at the \emph{same} narrow range of distances. 
In other words, the distance-uniformity property introduced by Alon et al. when studying the sum basic network creation game is present in the sum classic network creation game, too. 


(ii) If instead of considering the average degree, we focus on the number of non-bridge links bought for any player, we can prove that there exists a constant $D$ such that the average degree is upper bounded by $\max(R, \frac{2n}{\alpha}+6)$ for any \NE graph of diameter larger than $D$, where $R$ is a positive constant. 

\end{abstract}










\section{Introduction}
\label{Introduction} 

The starting point of this article is the seminal model introduced by Fabrikant et al. in \cite{Fe:03} which we call the \emph{sum classic network creation game}. This model can be viewed as a strategic game modelling Internet-like networks without central coordination. In this model the distinct agents, who can be thought as nodes in a graph, buy links of prefixed price $\alpha$ to the other agents in order to be connected in the resulting network of size $n$. We analyse the structure of the resulting equilibrium networks, which are called \emph{Nash equilibria} (\NE) as well as the concept of the \emph{Price of Anarchy} (\PoA). Nash equilibria are configurations such that every agent (or node) is not interested in deviating from his current strategy, and the \PoA can be seen as a measure of how the efficiency of the system degrades due to selfish behaviour of its agents. It was conjectured that the \PoA is constant for any $\alpha$ \cite{Demaineetal:07} and this conjecture, which we call the \emph{constant \PoA conjecture}, has been validated for the range $\alpha = O(n^{1-\delta_1})$ with $\delta_1 \leq 1/\log n$ \cite{Demaineetal:07} and for the range $\alpha >n(1+\delta_2)$ for any positive constant $\delta_2$, \cite{Alvarezetal3}. These results have been achieved following incremental steps after a large amount of new original papers. What happens in the remaining range of the parameter $\alpha$? The best upper bound known for the \PoA in the remaining range is $2^{O(\sqrt{\log n})}$ but apart from this, very little information is known regarding equilibria for the same range of $\alpha$. In this paper we show further topological properties of equilibria that can help to better understand what such networks look like and our contribution may help to shed light onto the constant \PoA conjecture as well.

Let us first define formally the model and related concepts.

\subsection{Model and definitions} 
 The \emph{sum classic network creation game} $\Gamma$ is defined by  a pair $\Gamma= (V,\alpha)$ where $V = \left\{1,2,....,n \right\}$  denotes the set of players and  $\alpha>0$ is a positive parameter. Each player $u\in V$ represents a node of an undirected graph and $\alpha$ represents the cost of establishing a link.
 
 A \emph{strategy} of a player $u$ of $\Gamma$  is a subset $s_u \subseteq V \setminus \left\{u \right\}$, the set of nodes for which player $u$ pays for establishing a link. A strategy profile for $\Gamma$ is a tuple $s=(s_1,\ldots,s_n)$ where $s_u$ is the strategy of player $u$, for each player $u\in V$. Let $\mathcal{S}$ be the set of all strategy profiles of $\Gamma$. Every strategy profile $s$ has associated a \emph{communication network} that is defined as the undirected graph $G[s] = (V, \left\{uv \mid v \in s_u \lor u \in s_v \right\})$. Notice that $uv$ denotes the undirected edge between $u$ and $v$. 
 
 
Let $d_G(u,v)$ be the distance in $G$ between $u$ and $v$. The cost associated to a player $u \in V$  in a strategy profile $s$ is defined by $c_u(s) =  \alpha |s_u| + D_{G[s]}(u)$ where   $D_G(u) = \sum_{v\in V, v \neq u} d_{G}(u,v)$ is the sum of the distances from the player $u$ to all the other players in $G$. As usual, the social cost of a strategy profile $s$ is defined by  $C(s)= \sum_{u \in V}{c_u(s)}$. 
 
 A Nash Equilibrium (\NE) is a strategy vector $s$ such that  for every player $u$ and every strategy vector $s'$ differing from $s$ in only the $u$ component, $s_u \neq s_u'$,  satisfies $c_u(s) \leq c_u(s')$. In a \NE $s$ no player has incentive to deviate individually her strategy since the cost difference $c_u(s')-c_u(s) \geq 0$. 
  Finally, let us denote by $\mathcal{E}$ the set of all \NE strategy profiles. The  price of anarchy (PoA) of $\Gamma$  is defined as $PoA= \max_{s \in \mathcal{E}} 
 C(s)/\min_{s \in \mathcal{S}}  C(s)$. 

It is worth observing that in  a \NE $s=(s_1,...,s_n)$ it never happens that $u\in s_v$ and $v\in s_u$, for any $u,v\in V$. Thus, if $s$ is a \NE, $s$ can be seen as an orientation of the edges of $G[s]$ where an arc from $u$ to $v$ is placed whenever $v \in s_u$. It is clear that a \NE $s$ induces a graph $G[s]$ that we call  \emph{NE graph} and we mostly omit the reference to such strategy profile $s$ when it is clear from context. However, notice that a graph $G$ can have different orientations. Hence, when we say that $G$ is a \NE graph we mean that $G$ is the outcome of a \NE strategy profile $s$, that is, $G=G[s]$. 

Given a graph $G$ we denote by $X \subseteq G$ the subgraph of $G$ induced by $V(X)$. In this way, given a graph  $G=G[s] = (V,E)$, a node $v \in V$,  and $X \subseteq G$, the \emph{outdegree of} $v$ in $X$ is defined as $deg_X^+(v)= | \left\{ u \in V(X) \mid u \in s_v\right\}| $ and the \emph{degree of} $v$ in $X$ as $deg_X(v) = |\left\{ u \in V(X) \mid  uv \in E \right\}|$.  Furthermore, the average degree of $X$ is defined as $deg(X)= \sum_{v \in V(X)}deg_X(v)/|V(X)|$.

For a distance index $r$ and a node $u \in V(G)$, $A_r(u)$ is defined as the set of nodes from $V(G)$ at distance exactly $r$ from $u$. 

Furthermore, recall that in  a connected graph $G=(V,E)$ a  vertex is a \emph{cut vertex} if its removal increases the number of connected components of $G$. A graph is biconnected if it has no cut vertices. We say that $H \subseteq G$ is a \emph{biconnected component} of $G$  if $H$ is a maximal biconnected subgraph of $G$. More specifically, $H$ is such that there is no other distinct biconnected subgraph of $G$ containing $H$ as a subgraph. Given a biconnected component $H$ of  $G$ and a node  $u \in V(H)$, we define $S(u)$ as the connected component containing $u$ in  the subgraph induced by the vertices $(V(G)\setminus V(H)) \cup \left\{u \right\}$. Notice that $S(u)$ denotes the set of all nodes $v$ in the connected component containing $u$ induced by $(V(G) \setminus V(H)) \cup \{u\}$ and then, every shortest path in $G$ from $v$ to any node $w \in V(H)$ goes through $u$. 

In the following sections we consider $G$ to be a \NE   and $H \subseteq G$, if it exists, a non-trivial biconnected component of $G$, that is, a biconnected component of $G$ of at least three distinct nodes. Then we use the abbreviations $d,d_H$ to refer to the diameter of $G$  and the diameter of $H$, respectively, (although $d_G(u,v)$ denotes the distance between $u,v$ in $G$). 
\vskip 10pt

\subsection{Historical overview} 
It is well known that the diameter of any \NE graph $G$ plus one unit serves as an upper bound for the \PoA of that graph \cite{Demaineetal:07}. This is the main result used in \cite{Demaineetal:07} to reach the conclusion that the \PoA is constant for the range $\alpha = O(n^{1-\delta_1})$ with $\delta_1 \geq 1/\log n$. Moreover, for the range $\alpha < 12n \log n$ it is shown in the same paper that the diameter of equilibria and thus the \PoA, is at most $2^{O\left(\sqrt{\log n} \right)}$. Subsequently, Alon et al. in \cite{AlonDHKL14} analyse an equilibrium notion that is related to the sum classic network creation game when we restrict only to single edge swaps. This new equilibrium concept is called the \emph{sum basic equilibria} and is defined as follows. An undirected graph $G$ is said to be a sum basic equilibrium if any node $v \in V(G)$ cannot perform an individual swap of any edge $vw$ for any other edge $vw'$ in such a way that it strictly decreases the average distance to all the other nodes from $v$.  In an attempt to improve the bounds obtained in the sum classic network creation game, Alon et al. propose to study the diameter of equilibria for this new equilibrium notion.  In the paper, they first show an upper bound  of $2^{O(\sqrt{\log n})}$, which is not better than the one for the sum classic network creation game. However, the authors introduce a novel concept, the so-called notion of \emph{distance-uniformity}, and they establish a connection between large-diameter sum basic equilibria and distance-uniform graphs that, at first glance, could lead to an improvement for the bound on the diameter. Roughly speaking, an undirected graph $G$ is said to be \emph{distance-uniform} (or \emph{distance-almost-uniform}) iff the majority of the nodes from the network are contained in  the same distance layer (or in the same two consecutive distance layers) for each node that we consider, respectively. The key details in this definition are that the property must be satisfied for every node in the network and, moreover, that the distance index of the layer is, for any node considered, the same in a distance-uniform graph or approximately the same (two consecutive values) in any distance-almost-uniform graph. More precisely, the authors prove that there exist exponents $x,y$ such that the $x$ and $y$ powers of sum basic equilibria are distance-uniform and distance-almost-uniform graphs, respectively. From this it is not hard to see that if the diameter of distance-uniform graphs was logarithmic then the diameter of sum basic equilibria would we polylogarithmic. Unfortunately, this statement, that distance-uniform graphs have logarithmic diameter, was proposed as a conjecture but was refuted later in \cite{LaLo:17}. The first main result we prove in this paper is precisely an analogous relationship between equilibria for the sum classic network creation game and distance-uniform graphs which can be stated in a very similar form as it is done in \cite{AlonDHKL14}. 


\vskip 5pt

Jointly with all these results, it is well-known that the \PoA of trees is less than $5$ \cite{Fe:03}. This  combined with the \PoA-diameter relation, the authors from \cite{Alvarezetal3} reach the conclusion that the \PoA is constant for the range $\alpha > n(1+\delta_2)$ by showing that the diameter of any biconnected component of a non-tree \NE graph is upper bounded by a constant. More specifically, the authors study the average degree of any non-trivial biconnected component $H$, noted as $deg(H)$, for the corresponding range of $\alpha$. In order to do so, they prove upper and lower bounds for the term $deg(H)$ which contradict each other when $n, \alpha$, the size of $H$ and the diameter of $H$ satisfy some conditions. From this fact the conclusion follows easily. This technique, the study of upper and lower bounds for the $deg(H)$, had already been used to prove that every \NE is a tree and thus the \PoA is constant in \cite{Mihalakmostly,Mihalaktree,Alvarezetal} for the ranges $\alpha > 273n, \alpha > 65n, \alpha> 17n$, respectively. Furthermore, from Theorem 4 in \cite{Albersetal:06} it can be deduced that, for a general $\alpha$, the average degree of $H$ is no greater than $2n/\alpha+4$ and in \cite{Alvarezetalconstant} it is shown that $deg_H^+(v)$ is upper bounded by a constant when $\alpha > n$. The second main result of this paper upper bounds the maximum number of links from $H$ that any player in an equilibrium can have.  As we shall see later, as a consequence of this, we prove a non-trivial upper bound for the number of non-bridge links bought by every player from any \NE graph for the same range of $\alpha$ having diameter larger than some constant.



\subsection{Our contribution}

We study key properties that provide new insights regarding how equilibria look like for the mysterious range of $\alpha$ where it is not known wether the \PoA is constant.

(i) In Sect. \ref{sec:distance-uniformity}, we see that for any  $\epsilon >0$ small enough there exists  a distance index $r_{\epsilon}$ and a quantity $x_{\epsilon} = (4\alpha/n)\epsilon^{-1} +1$ such that any node $u$ has at least $n(1-\epsilon)$ of the nodes from $G$ at a distance exactly in the interval $[ r_{\epsilon}-x_{\epsilon},r_{\epsilon}+x_{\epsilon}]$. In other words, the $2x_{\epsilon}$-th power of equilibria satisfy the notion of distance-uniformity introduced by Alon et al. for the sum basic network creation game.

(ii) In Sect. \ref{sec:ramsey} we show that there exists a constant $R$ such that the number of non-bridge links bought for any player in any equilibrium graph of diameter larger than a constant is at most $\max(R, \frac{2n}{\alpha}+6)$. More precisely, we show that any such equilibrium has at most one biconnected component $H$, and then we show that $deg_H^+(v) \leq \max(R, \frac{2n}{\alpha}+6)$, for any node $v\in V(H)$.  We prove this result in several steps. Our starting point is $G$ a \NE graph. In Subsect. \ref{subsec:biconnected}  we show that if $G$ has diameter larger than a constant it contains at most one biconnected component. Then, in Subsect. \ref{subsec:rough} we show a rough estimation $deg_H^+(v) = O((n/\alpha)^2)+O(1)$ for any biconnected component $H$ from $G$. Finally, in Subsect. \ref{subsec:precise}, using the previous estimation, we show the more precise estimation  $deg_H^+(v) \leq 2n/\alpha+6$ when the diameter of $G$ and $n/\alpha$ are larger than a constant for any biconnected component $H$ of $G$.

\section{Distance-uniformity and Nash equilibria}
\label{sec:distance-uniformity}

Given a graph $G$, recall that $A_i(u)$ is defined as the subset of nodes at distance exactly $i$ from $u \in V(G)$. Then, given $\epsilon > 0$, we say that a graph $G$ is \emph{$\epsilon-$distance-uniform} (or \emph{$\epsilon-$distance-almost-uniform}) iff there exists a distance index $r$ which we call the \emph{critical distance} such that $|A_r(u)| \geq n(1-\epsilon)$ (or $\max_{i \in \left\{ r,r+1\right\}}|A_i(u)| \geq n(1-\epsilon)$) for all $u \in V(G)$.  Intuitively, a distance-uniform graph is a graph for which \emph{every} node has at an approximately the \emph{same} distance the majority of the nodes from the graph. This concept was introduced by Alon et al. in \cite{AlonDHKL14} when studying the sum basic network creation game. In the same paper, the authors establish a connection between equilibria for the sum basic network creation game with distance-uniform graphs by showing that some non-trivial power of any sum basic equilibria is a distance-uniform graph: 

\begin{theorem}\label{thm:alon}\cite{AlonDHKL14} Any sum (basic) equilibrium graph $G$ with $n \geq 24$ vertices and diameter $d > 2 \log n$ induces an $\epsilon-$distance-almost-uniform graph $G'$ with $n$ vertices and diameter $\Theta(\epsilon d/ \log n)$ and an $\epsilon-$distance-uniform graph $G'$ with $n$ vertices and diameter $\Theta(\epsilon d / \log^2 n)$. 

\end{theorem}

Where the word \emph{induce} in the previous statement means that there exists an exponent $x$ such that $G'=G^x$. Now, let a \emph{buying deviation} be a deviation that consists only in buying links so that a \emph{buying \NE} is a \NE when restricting only to buying deviations.  Of course, every \NE in the general sense is, in particular, a buying \NE. Considering this we see that an analogous property to Theorem \ref{thm:alon} holds for equilibria in the sum classic network creation game.
In the following, we deduce an analogous property for the buying equilibria in the sum classic network creation game. 

More precisely, we see that for any  $\epsilon >0$ and any buying \NE $G$, there exists  a distance index $r_{\epsilon}$ and a quantity $x_{\epsilon} = (4\alpha/n)\epsilon^{-1} +1$ such that any node $u$ has at least $n(1-\epsilon)$ of the nodes from $G$ at a distance exactly in the interval $[ r_{\epsilon}-x_{\epsilon},r_{\epsilon}+x_{\epsilon}]$. In other words, the $2x_{\epsilon}$ power of equilibria satisfy the notion of distance-uniformity introduced by Alon et al. for the sum basic network creation game.    Therefore, it is important to notice that for the range $\alpha < 4n$, which is our range of interest since for $\alpha > 4n-13$ every \NE is a tree,  the quantity $2x_{\epsilon}$ is $O(\epsilon^{-1})$. Therefore, this is a significant improvement over the exponents $O(\epsilon^{-1}\log n)$ and $O(\epsilon^{-1}\log^2 n)$ which are considered in the proof of Theorem \ref{thm:alon} 

In order to achieve this goal we introduce some extra notation that allows us to simplify some calculations. Given $u, v \in V(G) $ and $i$ an integer value, we define $R^i_u(v)=\left\{z \mid -i \leq d_G(z,v) - d_G(u,v) \leq i  \right\}$.

\begin{figure}
\begin{center}
\includegraphics[scale=0.6]{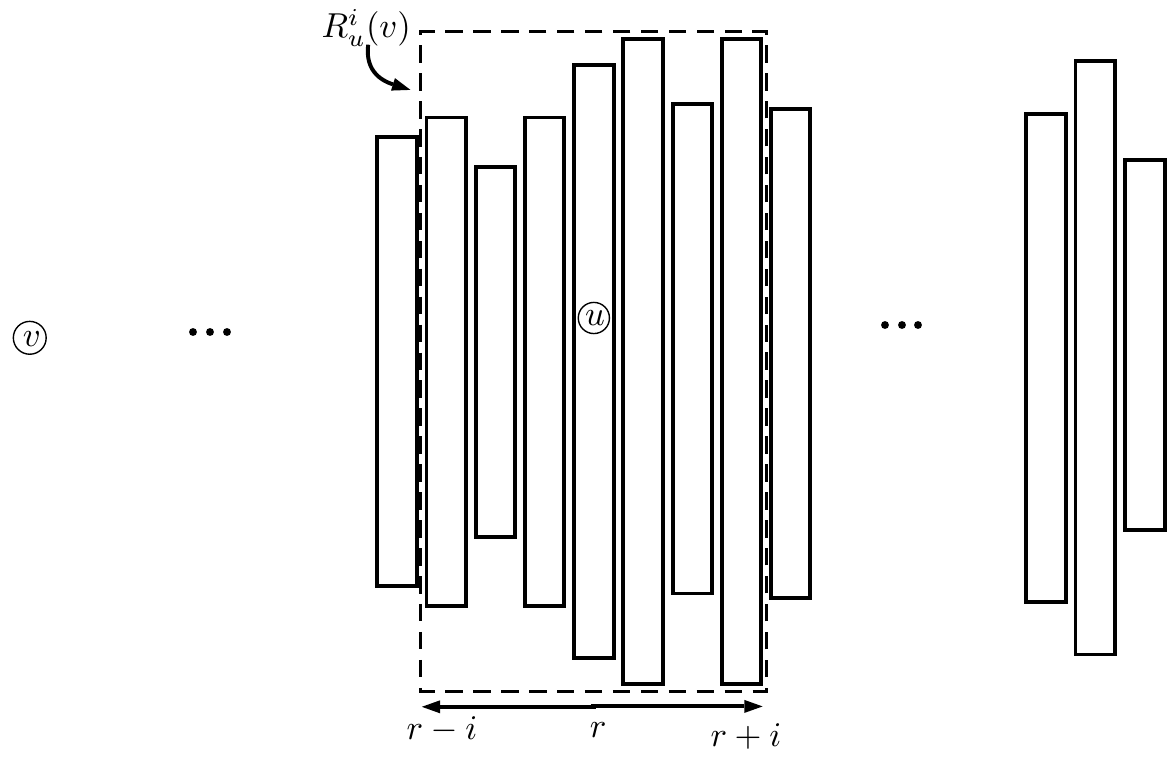}
\end{center}
\caption{The set $R_u^i(v)$ is the union of the $i$ distance layers to the left and to the right of the distance layer where $u$ belongs.}
\end{figure}

$R_u^i(v)$ can be seen as the distance layer with respect to $v$ which contains $u$ together with the corresponding $i$ distance layers to the left and to the right from this layer, although some of them could be empty if $d_G(u,v)+i>d$ or $d_G(u,v)-i<0$.

Furthermore, if $r$ is a distance index, then we define $R^i_r(u)=\cup_{-i \leq j \leq i} A_{r+j}(u)$. In other words, $R_r^i(v)$ is the $r$-th distance layer with respect to $v$ together with the corresponding $i$ distance layers to the left and to the right from this layer, although some of them could be empty if $r+i > d$ or $r-i < 0$.

Finally, given $v_1, v_2 \in V(G) $ and an integer $i$, we define $M_i(v_1,v_2)=\left\{z \mid -i \leq d_G(z,v_1) - d_G(z,v_2) \leq i  \right\}$. In other words, $M_i(v_1,v_2)$ is the set of nodes such that the difference of the distances to $v_1$ and $v_2$ in absolute value is no greater than $i$. For instance, $M_0(v_1,v_2)$ is the set of nodes equidistant from $v_1,v_2$ and $M_d(v_1,v_2) = V(G)$.  In the figure \ref{fig:med} it is assumed that $d_G(v_1,v_2)=k$ and the tag $[x_1,x_2]$ is used to describe the subset of nodes for which the distances to $v_1$ and $v_2$ are of the form $x_1$ and $x_2$, respectively. 

\begin{figure}
\label{fig:med}
\begin{center}
\includegraphics[scale=0.5]{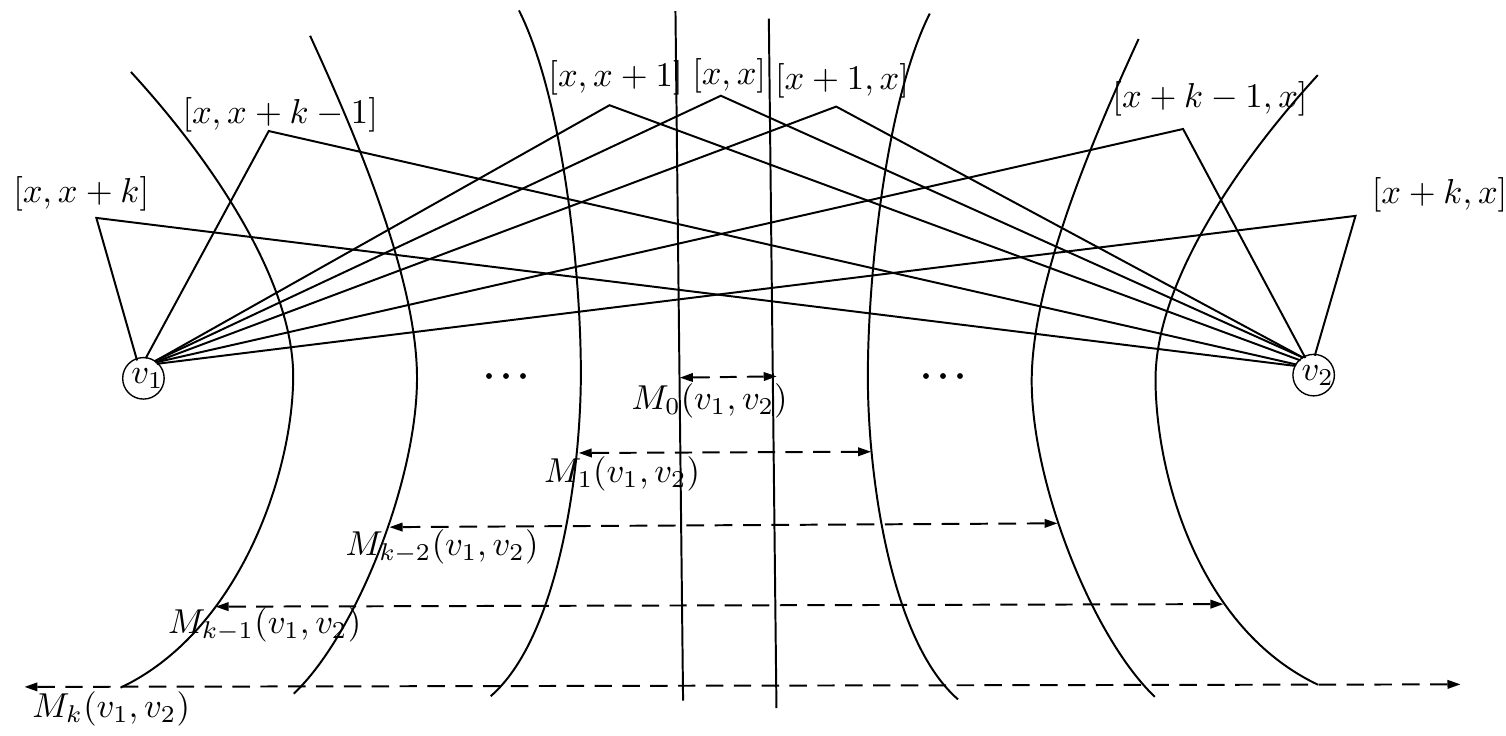}
\end{center}
\caption{The set $M_i(v_1,v_2)$ is the union of the nodes such that the difference of the distances to $v_1$ and $v_2$ is upper bounded by $i$.}
\end{figure}

\begin{proposition}
\label{prop:metric}
 Any buying NE graph $G$ satisfies that there exists a distance index $r$ such that $\sum_{i \geq 2}^d|R^i_r(u)| \geq nd-n-4\alpha$ for every node $u \in V(G)$, where $d=diam(G)$.
 
 \end{proposition}

\begin{proof} Consider a fixed node $u\in V(G)$. Let $S_{v}$ be the sum of the cost differences associated to the following two deviations: first, the deviation with respect to $u$ which consists in buying a link to $v$ and, in second place, the deviation with respect to $v$ which consists in buying a link to $u$.  Furthermore, let $\Delta_v(z)$ be the sum of the distance changes from $u$ to $z$ when $u$ buys a link to $v$ and from $v$ to $z$ when $v$ buys a link to $u$. Defining $S=\sum_{v\in V(Z)}S_v$, then we have have the relationship $S = 2\alpha(n-1)-\sum_{z\in V(G)} \sum_{v\in V(G)}\Delta_v(z)$. 

For any node $z \in V(G)$ we study now the sum $\sum_{v\in V(G)}\Delta_v(z)$. Clearly, $\sum_{v \in R_u^1(z)}\Delta_v(z) = 0$ and, in general: 

$$\sum_{v\in R_{u}^i(z) \setminus R_u^{i-1}(z)}\Delta_v(z)=\sum_{v\in A_{d(u,z)-i}(z) \cup A_{d(u,z)+i}(z)}\Delta_v(z)=i-1$$

For any $i \geq 2$.  This implies:

$$\sum_{v \in V(G)}\Delta_v(z) = \sum_{i \geq 2}^d(i-1)|R_{u}^i(z) \setminus R_u^{i-1}(z)|= \sum_{i \geq 1}^d(n-|R_u^i(z)|)$$

\noindent Using that $G$ is a \NE graph by hypothesis, this implies that $$0 \leq S = 2\alpha (n-1)-\sum_{z \in V(G)}\sum_{i \geq 1}^d(n-|R_u^i(z)|)$$

 From here  we deduce that there exists a node $v \in V(G)$ such that $\sum_{i \geq 1}^d |R_u^i(v)| \geq nd - 2\alpha$. Let $r = d_G(u,v)$ so that $|R_r^i(v)|=|R_u^i(v)|$.

On the other hand, let $S'$ be the sum of the cost differences associated to the deviations in $v_1,v_2$, respectively, that consist of buying a link from $v_1$ to $v_2$ and from $v_2$ to $v_1$, respectively. For any node $z \in V(G)$ let us study the sum of the distance changes with respect to these deviations. Clearly, the sum of the distance changes appearing in $S'$ that correspond to $z\in M_1(v_1,v_2)$ is $0$. In general, the sum of the distance changes appearing in $S'$ that correspond to $z\in M_i(v_1,v_2) \setminus M_{i-1}(v_1,v_2)$ is exactly $i-1$ for any $i \geq 2$. This implies that the sum of the distance changes appearing in $S$ is exactly:

$$\sum_{i \geq 2}^d(i-1)|M_i(v_1,v_2) \setminus M_{i-1}(v_1,v_2)|= \sum_{i \geq 1}^d(n-|M_i(v_1,v_2)|)$$

\noindent Using that $G$ is a \NE graph by hypothesis, this implies that:

 $$0 \leq S' = 2\alpha -\sum_{i \geq 1}^d(n-|M_i(v_1,v_2)|)$$ 

\noindent Then, $\sum_{j \geq 1}^d |M_j(v_1,v_2)| \geq nd - 2\alpha$.

\noindent Let $w \neq v$ be an arbitrary node from $V(G)$. If $z$ is such that $r-i \leq d_G(z,v) \leq r+i$ and $-j \leq d_G(z,v)-d_G(z,w) \leq j$ then $r-i-j \leq d_G(z,w) \leq r+i+j$. Therefore the following inclusion holds $R_r^i(v) \cap M_j(v,w) \subseteq R_r^{i+j}(w)$.

\vskip 5pt
\noindent Now, for each $z\in V(G)$, let $X(z)$ be the number of times that $z$ appears in the sequence of nested sets $R_r^{1}(v) \subseteq R_r^2(v) \subseteq ... \subseteq R_r^{d}(v)$. Similarly, let $Y(z)$ be the number of times that $z$ appears in the sequence of nested sets $M_1(w,v) \subseteq M_2(w,v) \subseteq ... \subseteq M_d(w,v)$. Finally, let $X'(z)$ be the number of times that $z$ appears in the sequence of nested sets $R_r^2(w) \subseteq R_r^3(w) \subseteq ... \subseteq R_r^d(w)$. By the previous inclusion relationship we have that $X'(z) \geq X(z)+Y(z)-d-1$. Therefore, adding all these inequalities for every $z \in V(G)$ we get: 

$$ \sum_{i \geq 2}^d |R_r^i(w)| = \sum_{z \in V(G)}X'(z) \geq \sum_{z \in V(G)}(X(z)+Y(z))-dn-n =$$

$$ =\sum_{i \geq 2}^d |R_r^i(v)| + \sum_{i \geq 1}^d |M_i(v,w)| -dn -n\geq dn-n-4\alpha$$

\end{proof}

As an immediate consequence we reach the main result of this section.

\begin{theorem}  \label{thm:formuladist}
For any  $\epsilon >0$ small enough there exists  a distance index $r_{\epsilon}$ and a quantity $x_{\epsilon} = (4\alpha/n)\epsilon^{-1} +1$ such that any node $u$ has at least $n(1-\epsilon)$ of the nodes from the network at a distance exactly in the interval $[ r_{\epsilon}-x_{\epsilon},r_{\epsilon}+x_{\epsilon}]$.
\end{theorem}

\begin{proof}

\noindent Using Proposition \ref{prop:metric}, we get that $\sum_{i \geq 2}|R_r^i(v)| \geq nd-n-4\alpha$, where $d=diam(G)$. Then for $d\geq x$ we obtain  $dn-4\alpha-n \leq  \sum_{i \geq 2}^d|R_r^i(v)| \leq (x-1) |R_r^{x}(v)|+(d-x)n$. And this implies that $|R_r^{x}(v)| \geq n\left(1-\frac{4\alpha}{n(x-1)} \right) $, which is equivalent to what we wanted to prove. 

\end{proof}

\begin{corollary}
For any $\epsilon > 0$ small enough there exists a quantity $x_{\epsilon} = (4\alpha/n)\epsilon^{-1} +1$ such that  the $2x_{\epsilon}$ power of any buying \NE graph is $\epsilon-$distance-almost-uniform. 
\end{corollary}

Finally, notice that if we consider the range $\alpha < n/C$ with $C > 4$, applying exactly the same reasoning of this corollary, we reach the conclusion of Proposition 10 from \cite{Alvarezetal} which is that every $4$th power of any \NE graph $G$ is a $4\alpha/n$-distance-almost-uniform graph.  


\section{The number of non-bridge links owned by every player}
\label{sec:ramsey}

Recall that a bridge from a graph is any edge such that when it is removed the number of connected components of the same graph increases.  Furthermore, if $diam(G) < D$ for every equilibrium graph $G$ then $\PoA \leq \max( D+1,5)$. From here, with these ideas in mind, in order to make further steps into proving the constant \PoA conjecture, it is enough if we pay attention to networks having diameter larger than a specific constant $D$ chosen conveniently. 

In this section we provide precisely a non-trivial upper bound for the number of non-bridge links bought by any player from equilibria having diameter larger than a constant $D$ to be determined in subsection 3.3. 

Our reasoning is divided into several steps:
 
 (1) We first show that there exist constants $D_1,R$ such that for any biconnected component $H$ of a non-tree equilibrium graph $G$ of diameter greater than $D_1$ it holds that $deg_H^+(v) \leq \max(R, \frac{2n}{\alpha}+6)$ (Theorem \ref{thm:degbound2-k} in subsection 3.2). To reach this result, we first prove an intermediate result which is that $deg_H^+(v) = O((n/\alpha)^2)+O(1)$ (Theorem \ref{thm:4} in subsection 3.1).
 
 \vskip 5pt
 
 (2) We show that any \NE graph $G$ having diameter larger than a constant $D_2$ has at most one biconnected component (Theorem \ref{thm:biconnected} in subsection 3.3).  Therefore, any \NE graph having diameter larger than $D_2$ consists of a unique biconnected component $H$ and then for every $v\in V(H)$, $S(v)$ is a tree.
 
 \vskip 5pt
 
 (3) Combining (1) and (2) and using the basic idea that in a graph that admits at most one biconnected component $H$ the number of non-bridge links owned by any player $v$ is either $0$ if $H$ does not exist or $deg_H^+(v)$ if $H$ exists, then we reach the next statement:

\begin{theorem}
\label{thm:deg}
There exist constants $R,D=\max(D_1,D_2)$ such that every player of any \NE graph $G$ having diameter larger than $D$ owns at most $\max(R, 2n/\alpha+6)$ non-bridge links.  
\end{theorem}

 
 \vskip 10pt

Now notice the following observations. The $(k,l)-$clique of stars introduced by Albers et al. in \cite{Albersetal:06}, is a \NE graph when $\alpha = l$.  Therefore, this construction shows that our upper bound on the maximum directed degree in $H$ is the best possible asymptotically speaking because taking one node $u$ from the $k-$clique component buying all the links to the remaining nodes from the $k-$clique then $deg_H^+(u) = k-1 = n/\alpha-1$. 

Notice that if $u$ is any node minimising the sum of distances on $V(H)$ having local diameter $d$, then, for any $v \in A_{d}(u) \cap V(H)$ having $k=deg_H^+(v)$ edges from $H$, we have that $0 \leq -(k-1)\alpha + n+ D(u)-D(v) \leq -(k-1)\alpha+n$, by considering the deviation in $v$ that consists in selling the $k$ edges from $H$ and buying a link to $u$. Hence, $deg_H^+(v) = k \leq n/\alpha+1$. However, we would like to remark that the same reasoning can not be applied, at least directly, if we pick any node $v\in V(H)$ for which every node $u \in V(H)$ at a maximal distance from $v$ verifies that its local diameter is strictly greater than the local diameter of $v$.

\subsection{Generalised $A$ sets.}
\label{subsec:rough}

Let $u$ be a prefixed node and suppose that we are given $v \in V(H)$ and $e_1=(v,v_1),...,e_k=(v,v_k)$ links bought by $v$. The \emph{$A$ set of $v,e_1,...,e_k$ with respect to $u$}, denoted as $A^u_{e_1,...,e_k}(v)$, is the subset of nodes $z \in V(G)$ such that every shortest path (in $G$) starting from $z$ and reaching $u$ goes through $v$ and the predecessor of $v$ in any such path is one of the nodes $v_1,...,v_k$.  

Therefore, notice that $v \not \in A^u_{e_1,...,e_k}(v)$. Then, for any $i =1,...,k$, we define the \emph{$A^{i}$ set of $v,e_1,...,e_k$ with respect to $u$}, noted as $A^{i,u}_{e_1,...,e_k}(v)$, the subset of nodes $z$ from $A^u_{e_1,...,e_k}(v)$ for which there exists a shortest path (in $G$) starting from $z$ and reaching $u$ such that goes through $v$ and the predecessor of $v$ in such path is $v_i$. 

Notice that, $A^u_{e_1,...,e_k}(v) = A^{1,u}_{e_1,...,e_k}(v) \cup  ... \cup A^{k,u}_{e_1,...,e_k}(v)$ and $A^{i,u}_{e_1,...,e_k}(v)= \emptyset$ iff $d_G(u,v_i) = d_G(u,v)-1$ or $d_G(u,v_i)=d_G(u,v)$.  Furhtermore, the subgraph induced by $A^{i,u}_{e_1,...,e_k}(v)$ is connected whenever $A^{i,u}_{e_1,...,e_k}(v) \neq \emptyset$. 

Define $crossings(X,Y)$ for subsets of nodes $X,Y \subseteq V(G)$ to be the set of edges $xy$ with $x\in X$, $y \in Y$. 

 Then notice that the following simple property holds:

\begin{remark} \label{remark:surface}Let $xy$ be a crossing between $A_{e_1,...,e_l}^{j,u}(v)$ and its complement. Then $x,y\in V(H)$.  
\end{remark}

\begin{proof} For any $z \in V(G)$ let us use the notation $z'$ to refer to the unique node from $V(H)$ for which $z \in S(z')$ and let $X =A_{e_1,...,e_l}^{u,j}(v)$. 
First we claim that when $v\in V(H)$ (which is the case under consideration) for any two nodes $z_1,z_2 \in S(z')$ either $z_1,z_2 \in X$ or  $z_1,z_2 \not \in X$. This is because if $z\in X$ then implies $z' \in X$, too, by definition of the $A$ sets and by the definition of cut vertex and, conversely, if $z'\in V(H)$ and $z' \in X$ then $S(z') \subseteq X$ by the definition of the $A$ sets and by the definition of cut vertex, too. 

Now, if for the sake of the contradiction some crossing $xy \in crossings(X,X^c)$ satisfies that $x \not \in V(H)$, then this implies that $y \not \in S(x')$ because by the previous observation there cannot exist two distinct nodes $z_1,z_2 \in S(z')$ with $z_1 \in X$ and $z_2 \in X$. Then, there would exist two distinct paths from $x \in S(x')$ to $y \not \in S(x')$. The first one, would start from $x$ and would go to $x'$ then $y'$ and finally to $y$. The second one (distinct from the first one) would directly use the edge $xy$ to go from $x$ to $y$. This contradicts the definition of cut vertex and thus we have reached a contradiction.

\end{proof}

Then, given $l \leq k$ let $Z_{e_1,...,e_l}$ be the graph having nodes $z_1,...,z_l$ with an undirected edge between $z_i,z_j$ with $i \neq j$ iff $crossings(A^{i,u}_{e_1,...,e_k}(v), A^{j,u}_{e_1,...,e_k}(v)) \neq \emptyset$. Intuitively, $Z_{e_1,...,e_l}$ is the abstract graph of the crossing relations between the $A^i-$sets $A^{i,u}_{e_1,...,e_k}$ defined by $e_i \in \left\{ e_1,...,e_l\right\}$.

\begin{figure}
\begin{center}
\includegraphics[scale=0.5]{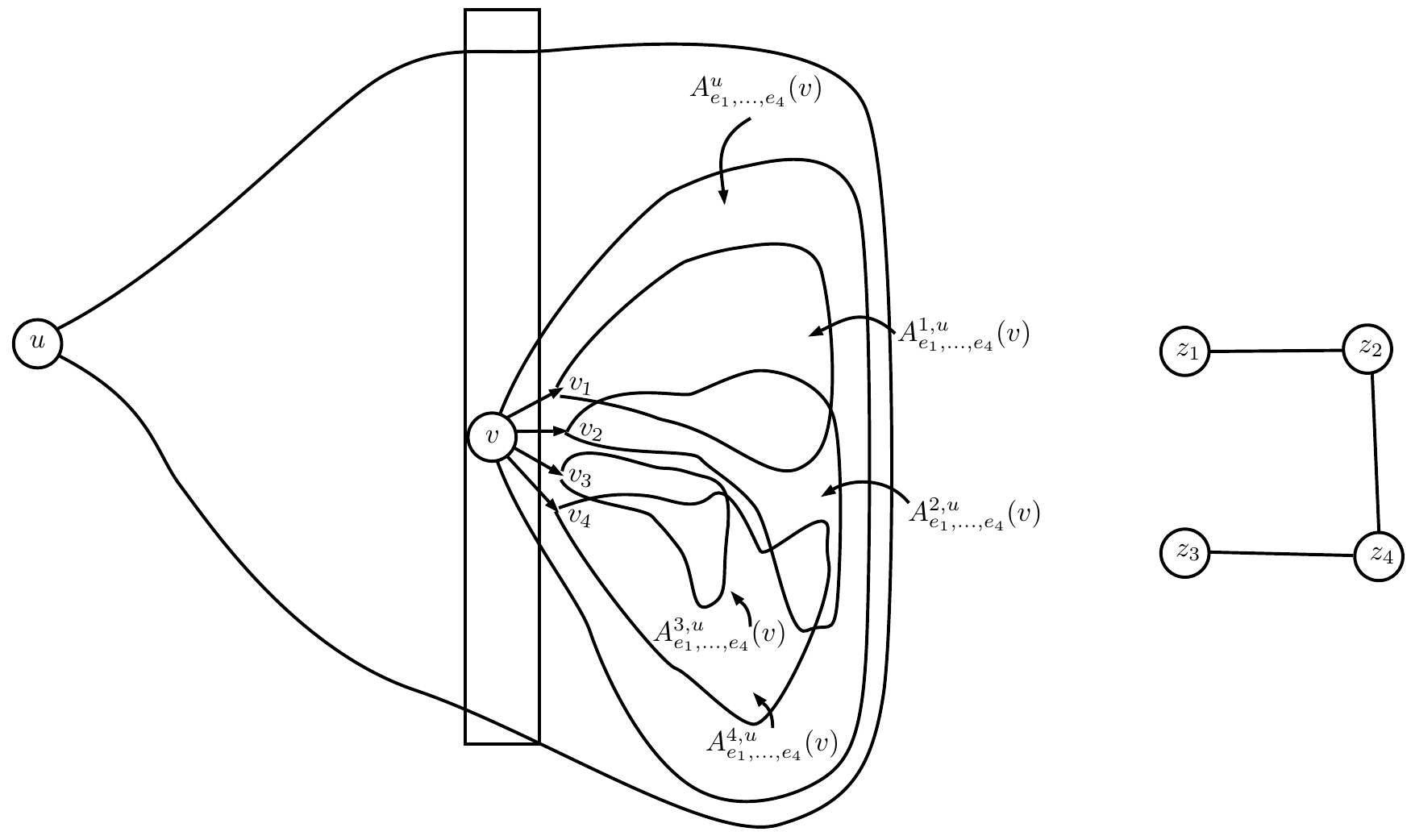}
\end{center}
\caption{On the left, the $A$ and $A^i$ set of $v$ with respect to $u$ and $e_1,...,e_4$. On the right, the graph $Z_{e_1,...,e_4}^4$.}
\end{figure}

Let $M_{e_{1},...,e_{l}}$ denote the maximum diameter of the connected components of the graph $Z_{e_{1},...,e_{l}}$. Then if we consider $\Delta C$ the cost difference associated to the deviation in $v$ that consists in deleting $e_1,...,e_l$ and buying a link to $u$, we can generalise and unify the formulae from Proposition 1 and Proposition 2 from \cite{Alvarezetal3} by stating and proving the following result: 

\begin{proposition}
\label{prop:formulageneral} 
Let $u, v \in V(H)$ and let $e_1,...,e_k \in E(H)$ be links bought by $v$. Let $\Delta C$ be the cost difference of the deviation on $v$ that consists in deleting $e_1,...,e_l$ with $l \leq k$ and buying a link to $u$. Then: 

$$\Delta C \leq -(l-1)\alpha+n+D(u)-D(v)+2d_H|A^u_{e_1,...,e_k}(v)|(1+M_{e_{1},..,e_{l}})$$
\end{proposition}

\begin{proof}

The term $-(l-1)\alpha$ is clear because we are deleting the $l$ edges $e_{1},...,e_{l}$ and buying a link to $u$. Now let us analyse the difference of the sum of distances in the deviated graph $G'$ vs the original graph. For this purpose, let $z$ be any node from $G$. We distinguish two cases:

(A) If $z \not \in A^u_{e_{1},...,e_{l}}(v)$ then:

(1) Starting at $v$, follow the connection  $vu$.

(2) Follow a shortest path from $u$ to $z$ in the original graph. 

In this case we have that:  $d_{G'}(v,z) \leq 1+d_G(u,z)$. 

(B) If $z \in A^u_{e_{1},...,e_{l}}(v)$ then there exists some index $j$ such that $z \in A^{j,u}_{e_{1},...,e_{l}}(v)$. For this index $j$, let $CC_{j}=\left\{ z_{i_1},...,z_{i_s}\right\}$ be the connected component from $Z_{e_{1},...,e_{l}}$ where $z_{j}$ belongs. Since $H$ is biconnected there must exists an index $j_0$ with $z_{j_0} \in CC_{j}$ such that there exists a crossing $xy$ between the set $A^{j_0,u}_{e_1,...,e_l}(v)$ and the complement of $\cup_{h=1}^sA^{i_h,u}_{e_1,...,e_l}(v)$ with $x \in A^{j_0,u}_{e_1,...,e_l}$, and thus $x\in V(H)$ by Remark \ref{remark:surface}.  Finally, suppose that $z_{j_0}-z_{j_1}-...-z_{j_M} = z_j$ is a shortest path connecting $z_{j_0}$ with $z_j$ inside $CC_{j}$ for which we have $M\leq M_{e_1,...,e_l}$.  Now consider the following path in $G$:

\begin{figure}
\begin{center}
\includegraphics[scale=0.5]{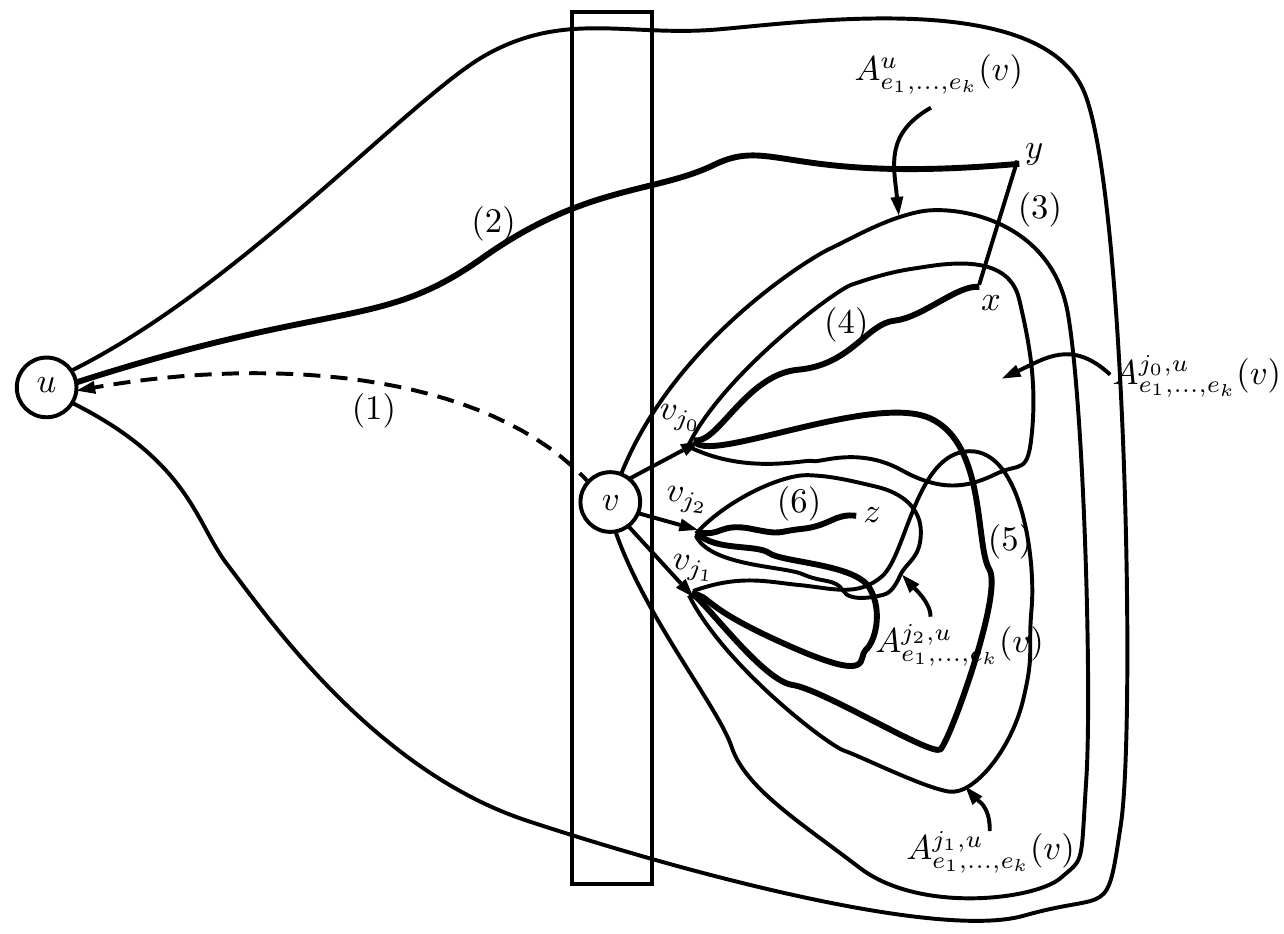}
\end{center}
\caption{The alternative path that $z$ could use to reach $v$ in the deviated graph.}
\end{figure}

 (1) Starting at $v$, follow the connection $vu$, which corresponds to one unit distance.

 (2) Follow any path from $u$ to $y$ which is contained in the complement of $\cup_{h=1}^sA^{i_h,u}_{e_1,...,e_l}(v)$ by the definition. Then we have $d_G(u,y)\leq d_G(u,v)+d_G(v,x)+1$. Therefore, in this case we count at most $d_G(u,v)+d_G(v,x)+1$ unit distances. 
		
(3) Cross the connection  $yx$, which corresponds to one unit distance.

(4) Go from $x$ to $v_{j_0}$ inside $A^{j_0,u}_{e_1,...,e_l}(v)$ giving exactly $d_G(x,v_{j_0})-1$ unit distances.

(5) Go from $v_{j_0}$ to $v_j=v_{j_t}$ inside $A^u_{e_1,...,e_l}(v)$. The length of the path inside $A^u_{e_1,...,e_l}(v)$ is no larger than $2d_HM$ since moving from $v_{j_{h}}$ to $v_{j_{h+1}}$ with $0 \leq h \leq M-1$ gives at most $2d_H$ distance units at each step. 

(6) Go from $v_j$ to $z$ inside $A^{j,u}_{e_1,...,e_l}(v)$ giving exactly $d_G(v_j ,z)-1$ unit distances.


In this case we have that:
\begin{align*}
d_{G'}(v,z) \leq & \overbrace{1}^{(1)} + \overbrace{d_G(u,v)+d_G(v,x)+1}^{(2)}+\overbrace{1}^{(3)}+\overbrace{d_G(x,v_{j_0})-1}^{(4)}+\overbrace{2d_HM}^{(5)}+\overbrace{d_G(v_j,z)-1}^{(6)}  \\
 =& 1+d_G(u,z)+(d_{G}(v,x)-1)+2d_HM+d_G(x,v_{j_0}) < 1+d_G(u,z)+2(M_{e_1,...,e_l}+1)d_H
\end{align*}
Combining the two inequalities we reach the conclusion: 
$$\Delta C \leq -(l-1)\alpha +n+D_G(u)-D_G(v)+2d_H(M_{e_1,...,e_l}+1)|A^u_{e_1,...,e_k}(v)|$$

\end{proof}

Recall that an \emph{independent set} from a given graph is a subset of nodes such that there is no edge between any pair of them.  Furthermore, the \emph{independence number} of any given graph is the maximum size of any of its independent sets. Recall the following well-known result: 

\begin{lemma} \cite{Wei1} \label{alfa(G)}The independence number of any graph $Z$ is greater than or equal $\sum_{v \in V(Z)} \frac{1}{1+deg_Z(v)}$
\end{lemma}

We will use lemma \ref{alfa(G)} in combination with the formula obtained in Proposition \ref{prop:formulageneral} to reach the main result of this subsection. The idea is that we can choose distinct starting subsets of edges $e_1,...,e_l$ and then, after completing it to the subset $e_1,...,e_k$, we can invoke Proposition \ref{prop:formulageneral}. When we combine the formula given in the proposition with simple other inequalities obtained from buying deviations, we get information about the number $l$ and the number $M_{e_1,...,e_l}$. Specifically, we end up deducing an upper bound for the cardinality of a maximum independent set of $Z_{e_1,...,e_k}$ and an upper bound for the maximum degree of any node from $Z_{e_1,...,e_k}$, too.  

\begin{theorem} 
\label{thm:4}
There exist constants $R_1$ and $R_2$ such that for any node $v \in V(H)$ we have $deg_H^+(v) \leq \max(R_1, R_2n^2/\alpha^2)$. 
\end{theorem}

\begin{proof}

We distinguish between the cases $d_H > 2$ or $d_H \leq 2$. 

First, suppose that $d_H > 2$. Suppose that $v \in V(H)$ has bought $k$ edges from $H$ and let $e_1,...,e_l$ be any $l$-size subset of such edges, assuming that $e_{l+1},...,e_k$ are the remaining ones. Also, pick $u \in V(H)$ any node at distance $d_H/2$ from $v$, which always exists. By Proposition \ref{prop:formulageneral} we know that if we consider the deviation in $v$ that consists in deleting $e_{1},...,e_{l}$ and buying a link to $u$ we get that if $\Delta C$ is the corresponding cost difference: 

\begin{equation}
\label{eq:1}\Delta C \leq -(l-1)\alpha +n+D_G(u)-D_G(v)+2d_H(M_{e_{1},...,e_{l}}+1)|A^u_{e_1,...,e_k}(v)| 
\end{equation}

Now consider the deviation in $u$ that consists in buying a link to $v$. Let $\Delta C_{buy}$ be the corresponding cost difference. Then:

\begin{equation}
\label{eq:2}\Delta C_{buy} \leq \alpha + n+ D_G(v)-D_G(u)
\end{equation}

By the definition of the $A$ sets, another inequality can be obtained considering the same deviation:

$$\Delta C_{buy} \leq \alpha - (d_G(u,v)-1)|A^u_{e_1,...,e_k}(v)|$$

Imposing that $G$ is a \NE and using that $d_H > 2$ we now deduce:  

 \begin{equation}
\label{eq:4} 2d_H |A^u_{e_1,...,e_k}(v)| \leq 2d_H\frac{\alpha}{d_G(u,v)-1} = 4 d_H\alpha /(d_H-2) \alpha  \leq 12\alpha
 \end{equation}

 Therefore, combining (\ref{eq:1}), (\ref{eq:2}) and (\ref{eq:4}) we obtain: 

$$\Delta C+ \Delta C_{buy} \leq -(l-2)\alpha+2n+12\alpha + 12\alpha M_{e_{1},...,e_{l}}$$

Imposing that $G$ is a \NE, then  

$$l \leq \frac{2n}{\alpha}+14 + 12M_{e_{1},...,e_{l}}$$

Now that we have reached this result consider the following two observations. In the first place, if we assume that $z_2,...,z_{l}$ are exactly all the neighbours of $z_1$ from $Z_{e_1,...,e_k}$ then $M_{e_1,...,e_l}\leq 2$ and we obtain an upper bound on the maximum degree in any node in $Z_{e_1,...,e_k}$:  $ l \leq \frac{2n}{\alpha}+38$. 

Secondly, if we assume that $z_1,...,z_m$ conform an independent subset of nodes from $Z_{e_1,...,e_k}$ then $M_{e_1,...,e_m}=0$ and we obtain an upper bound on the independence number of $Z_{e_1,..,e_k}$:   $m \leq \frac{2n}{\alpha}+14$.

\noindent Therefore, from these two observations together with Lemma \ref{alfa(G)}  we deduce: 

$$\frac{2n}{\alpha}+14 \geq k\frac{1}{38+\frac{2n}{\alpha}}$$

\noindent From here the conclusion. 

\vskip 10pt

Now consider the case $d_H \leq 2$. We claim that if $d_H \leq 2$ then the number of links from $H$ bought by any player is at most $3n/\alpha$.

Let us suppose that $v\in V(H)$ has bought the edges $e_1,...,e_k$.  Then, deleting an edge $e_i \in E(H)$ gives a cost difference $\Delta C_{delete}$ that satisfies: 

$$\Delta C_{delete} \leq -\alpha + (2d_H -1)|A^{v}_{e_i}(v)| \leq -\alpha+3|A^v_{e_i}(v)|$$ 

Imposing that $G$ is a \NE graph we obtain $\alpha \leq 3 |A^v_{e_i}(v)|$. But the subsets $A^v_{e_i}(v)$ are mutually disjoint for distinct subindexes $i=1,...,k$. Therefore, adding all the corresponding inequalities we obtain that $k \alpha \leq 3n$ implying $k \leq \frac{3n}{\alpha}$.

\end{proof}

\subsection{The number of edges in $H$ owned by a player.}

\label{subsec:precise}

\noindent Let $H$ be a biconnected component of any equilibrium graph $G$. Now we can use the result stated in Theorem \ref{thm:4} to give an improved upper bound for the directed degree in $H$. More precisely, we see that in any non-tree \NE graph $G$ of diameter larger than some constant $D_1$, the number of links bought by any player from any of its biconnected components $H$ is at most $2 \frac{n}{\alpha}+6$ provided that $n/\alpha$ is larger than some constant $K$.

Before stating the main result recall the following remark:
 
\begin{remark}
\label{rem:variationdistances}
Let $v \in V(H)$ with $deg_H^+(v) = l$. Then, when removing $v$ and all the directed edges bought by $v$, the distance between any two nodes $z_1,z_2 \in V(G)$ with $z_1,z_2 \neq v$ increases in strictly less than $2d_Hl$ distance units. 
\end{remark}

\begin{proof}
Let $u = z_1$ and let $e_1,...,e_l$ be the $l$ edges bought by $v$ in $H$. Look at the proof of Proposition \ref{prop:formulageneral}.  If we set $z=z_2$, the distance between $u$ and $z$ in the graph obtained removing the edges $e_1,..,e_l$, is at strictly less than $d_G(u,z)+2d_H(1+M_{e_1,...,e_l})$. Since $M_{e_1,...,e_l}$ is at most $l-1$ then we obtain that the variation of distances is strictly less than $2d_Hl$, as we wanted to see.

\end{proof}

\begin{theorem}\label{thm:degbound2-k}
There exist constants $D_1,K$ such that any equilibrium graph of diameter greater than $D_1$ for $n/\alpha > K$, satisfies that $deg^+_H(v) \leq 2\frac{n}{\alpha}+6$ for every $v \in V(H)$. 
\end{theorem}
\begin{proof}
Let $v$ be any node from $V(H)$ and $(v,v_1),...,(v,v_l) \in E(H)$ $l$ distinct links from $H$ bought by $v$. Now let $k=O(1)$ be a positive integer. For any $k-$tuple $(u_1,...,u_k)$ of nodes (not necessarily distinct) from $V(G)$ let $G_{u_1,...,u_k}$ be the graph in which $v$ deletes the $l$ links $(v,v_1),...,(v,v_l)$ and buys $k$ links to $u_1,...,u_k$.  Let $X$ be the set of all $k-$tuples $(u_1,...,u_k)$ of nodes from $V(G)$ such that $G_{u_1,...,u_k}$ is connected. 

Now, for every $k-$tuple $(u_1,...,u_k)$ from $X$ consider the deviation that consists in deleting $(v,v_1),...,(v,v_l)$ and buying $k$ links to $u_1,...,u_k$. Furthermore, for every $i=1,...,k$ consider the deviation in $u_i$ that consists in buying a link to $v$. Let $S_{u_1,...,u_k}$ be the sum of the corresponding cost differences associated to these $k+1$ deviations.  

First, the term corresponding to the creation cost in $S_{u_1,...,u_k}$ is clearly $\alpha(2k-l)$. Let us now evaluate the term corresponding to the variation of distances in $S_{u_1,...,u_k}$. To this purpose, pick any node $u \in V(G)$. If $u \in S(v)$ then the variation of distances relative to $S_{u_1,...,u_k}$ is zero. Otherwise, let $u \in V(G) \setminus  S(v)$ and let $A^u(v)$ be the set of nodes $z$ such that every shortest path from $z$ to $u$ goes through $v$. We now distinguish the two cases depending on whether it holds $\left\{u_1,...,u_k\right\} \cap (A^u(v))^c = \emptyset$.


(i) First, suppose that there exists $i\in \left\{1,...,k\right\}$ such that $u_i \not \in A^u(v)$. On the one hand, there exists a path in $G_{u_1,...,u_k}$ from $u_i$ to $u$ that does not go through $v$ so the distance from $v$ to $u$ in $G_{u_1,...,u_k}$ is at most $1+d_G(u_i,u)$. On the other hand, when buying the link from $u_i$ to $v$ the distance from $u_i$ to $u$ in the deviated graph is at most $1+d_G(u,v)$. Furthermore, the distances between $u_j$ and $u$ with $j \neq i$ are less than or equal the distances in the original graph the deviations on such nodes consist in only buying a link. Therefore, the sum of the distance changes is at most 

$$(1+d_G(u_i,u)-d_G(u,v))+(1+d_G(u,v)-d_G(u,u_i)) =2$$

(ii) Otherwise, suppose that  $u_1,...,u_k \in A^u(v)$.

By Remark \ref{rem:variationdistances} the distance change from $v$ to $u$ relative to $S_{u_1,...,u_k}$ is at most $2dl$. Furthermore, the distance changes for the nodes $u_i$ is non-positive because the deviations on such nodes consist in only buying a link. Therefore, the sum of the distance changes is at most $2dl$ in this case. 

\begin{center}
\includegraphics[scale=0.45]{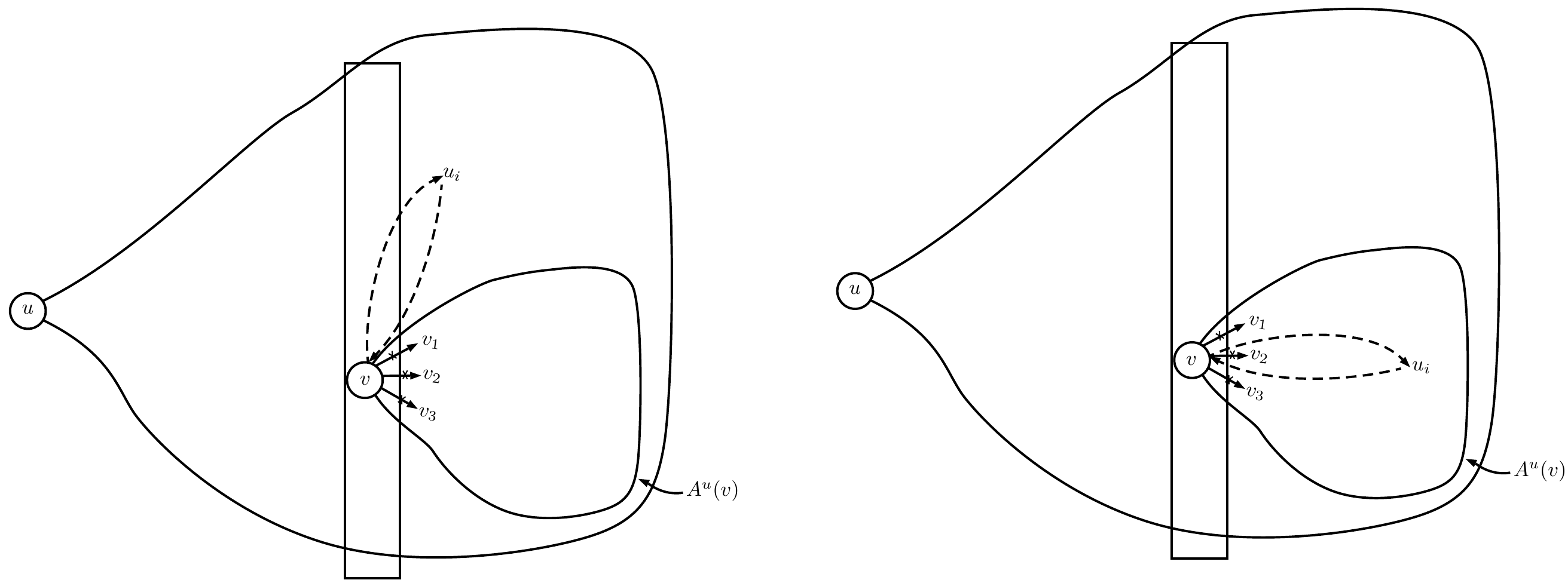}
\end{center}

Now let us evaluate the following sum $S = \sum_{(u_1,...,u_k) \in X} S_{u_1,...,u_k}$. By the previous observations: 

$$S \leq |X|\alpha(2k-l) + \left( 2n|X| +\sum_{u \in V(G)}2dl(|A^u(v)|^k-|S(v)|^k)\right) $$

Where the term $-|S(v)|^k$ comes from the fact that if all $u_i$ are in $S(v)$ then the deviated graph would be disjoint. Imposing that $G$ is a \NE, then: 

$$l \leq 2k+\frac{2n}{\alpha}+\frac{2dl}{\alpha}\sum_{u \in V(G)}\frac{|A^u(v)|^k-|S(v)|^k}{|X|} $$

Clearly, since $H$ is biconnected and $v\in V(H)$ then $G_{u_1,...,u_k}$ could be disconnected only when $u_1,...,u_k\in S(v)$ and thus $|X| \geq n^k-|S(v)|^k$. 
Therefore:

$$ \frac{|A^u(v)|^k-|S(v)|^k}{|X|} \leq \frac{|A^u(v)|^k-|S(v)|^k}{n^k - |S(v)|^k} \leq \frac{|A^u(v)|^k}{n^k} $$

Then it is enough if we show that the expression $ \frac{2dl}{\alpha} \sum_{u \in V(G)}\frac{|A^u(v)|^k}{n^k}$ can be made as small as we want. We distinguish between the cases $d_G(u,v) >1$ and $d_G(u,v) = 1$:

\vskip 5pt

(A) If $d_G(u,v) > 1$ then we know that $|A^u(v)| \leq \frac{\alpha}{d_G(u,v)-1}$ because if $u$ buys a link to $v$ he improves in $d_G(u,v)-1$ distance units for every node from $A^u(v)$. Therefore, using the notation $A_i(u)$ for the subset of nodes from $G$ at distance $i$ from $u$, we get: 

$$\frac{2dl}{\alpha} \sum_{ \left\{u \mid d_G(u,v)> 1\right\} } \frac{|A^u(v)|^k}{n^k} \leq \frac{2dl}{\alpha} \sum_{i \geq 2} \frac{\alpha^k}{n^k(i-1)^k}|A_i(v)|$$

Now, depending on whether $i$ is such that $2 \leq i \leq d/8-1$ or $d/8 \leq i$ we obtain that  $ \frac{\alpha^k}{n^k(i-1)^k} \leq  \frac{\alpha^k}{n^k}$ or  $ \frac{\alpha^k}{n^k(i-1)^k} \leq  \frac{\alpha^k}{n^k(d/8-1)^k}$, respectively. Therefore, defining $B_i(v) = \cup_{j \leq i }A_j(v)$, we obtain: 

$$ \frac{2dl}{\alpha} \sum_{i \geq 2} \frac{\alpha^k}{n^k(i-1)^k}|A_i(v)| < \frac{2dl}{\alpha}  \frac{\alpha^k}{n^k}|B_{d/8-1}(v)|+\frac{2dl}{\alpha} \frac{\alpha^k}{(d/8-2)^kn^k}(n-|B_{d/8-1}(v)|) $$

We have that $\frac{d}{(d/8-1)^k} = O(1/d^{k-1})$. Choosing $K$ to be large enough, if $n/\alpha>K$ by hypothesis then $l = O(n^2/\alpha^2)$ using Theorem \ref{thm:4}. Finally, $|B_{d/8-1}(v)| \leq 2\alpha/d$ because otherwise any node at distance at least $d/2$ from $v$, which always exists, would have incentive to buy a link to $v$. With these results then: 

$$ \frac{2dl}{\alpha} \sum_{ \left\{u \mid d_G(u,v)> 1\right\} } \frac{|A^u(v)|^k}{n^k}  = O\left(\frac{\alpha^{k-2}}{n^{k-2}} \right)+  O\left( \frac{1}{d^{k-1}} \frac{\alpha^{k-3}}{n^{k-3}} \right)$$

Therefore, since $k=O(1)$, for any small enough positive constant $\delta$ there exist quantities $K_1=K_1(\delta)$ and $D_1=D_1(\delta)$ such that if $n/\alpha > K_1$ and $d > D_1$ then  $\frac{2dl}{\alpha} \sum_{ \left\{u \mid d_G(u,v)> 1\right\} } \frac{|A^u(v)|^k}{n^k}\leq \delta$ and moreover, $K_1,D_1=O(1)$, thus showing what we wanted to prove.

\vskip 5pt

(B) Now suppose that $d_G(u,v)=1$. If $u \not \in \left\{ v_1,...,v_l\right\}$ then there is no increment in distance units. Otherwise, for each $v_i$ with $1 \leq i \leq l$, $|A^{v_i}(v)| \leq n$. Choosing $K$ to be large enough, if $n/\alpha>K$ by hypothesis then $l = O(n^2/\alpha^2)$ using Theorem \ref{thm:4}, and therefore: 

$$\frac{2dl}{\alpha} \sum_{ \left\{u \mid d_G(u,v)= 1\right\} } \frac{|A^u(v)|^k}{n^k} \leq \frac{2dl^2}{\alpha} = O(dn^4/\alpha^5)$$

Let $\delta > 0$ any small enough positive constant. First suppose that $\alpha > n^{\frac{4}{5}+\delta}$. Then, using the relation $d = 2^{O(\sqrt{\log n})}$ from \cite{Demaineetal:07} we deduce that:

$$\frac{2dl}{\alpha} \sum_{ \left\{u \mid d_G(u,v)= 1\right\} } \frac{|A^u(v)|^k}{n^k} \leq \frac{d}{n^{5\delta}} = o(1)$$

On the other hand, if $\alpha \leq n^{\frac{4}{5}+\delta}$ with $\delta$ small enough then we are in the range $\alpha = O(n^{1-\delta_1})$ with $n$ large enough so that $\delta_1 \geq 1/\log n$ and therefore $d = O(1)$ by the main result from \cite{Demaineetal:07}. 

\vskip 5pt

Therefore, gathering together these bounds and choosing $k=3$ we reach the conclusion that, there exists constants $D_1,K$ such that if $d>D_1$ and $n/\alpha>K$, then the expression $ \frac{2dl}{\alpha} \sum_{u \in V(G)}\frac{|A^u(v)|^k}{n^k}$ can be made sufficiently small in both cases (A) and (B). 

\end{proof}

\begin{theorem} There exist constants $R,D_1$ such that any equilibrium graph of diameter greater than $D_1$ satisfies  $\deg_H^+(v) \leq \max(R, \frac{2n}{\alpha}+6)$ for every $v\in V(H)$.
\end{theorem}

\begin{proof}
Let $K,D_1$ be the constants for which Theorem \ref{thm:degbound2-k} holds if $d=diam(G) > D_1$ and $n/\alpha > K$. Suppose that $H$ is a biconnected component of $G$. If the diameter $d=diam(G) > D_1$ then it holds that either $deg_H^+(v) \leq 2n/\alpha+6$ if $n/\alpha > K$ by Theorem \ref{thm:degbound2-k} or, otherwise, $n/\alpha \leq K$ and then $deg_H^+(v) \leq \max(R_1n^2/\alpha^2,R_2) < \max(R_1K^2,R_2)$ by Theorem \ref{thm:4}.  
\end{proof}

\subsection{Unique non-trivial biconnected component}
\label{subsec:biconnected}

In \cite{Mihalakswap} the authors show that any Assymetric-Swap Equilibrium graph has at most one non-trivial $2-$edge-connected component. In the following, we show that in fact, every \NE graph having diameter larger than a constant $D_2$ has at most one non-trivial biconnected component.  

To this purpose we distinguish two cases depending on the value of $\alpha$. 

\begin{lemma}\label{lem:biconnected1}
For $\alpha < \frac{n-1}{2}$ any equilibrium $G$ has at most one non-trivial biconnected component.
\end{lemma}

\begin{proof}
Let us suppose the contrary. Then there exists a cut vertex $v$ and subgraphs $H_1,H_2$ with the following properties: 

(i) $V(G)$ decomposes as  $V(H_1) \cup V(H_2) \cup \left\{ v\right\}$.

(ii) $H_1 \cup \left\{v \right\}$ and $H_2 \cup \left\{ v \right\}$ contain each of them a distinct biconnected component.

(iii) Every shortest path from any two nodes $w_1 \in V(H_1),w_2 \in V(H_2)$ goes through $v$. 

Clearly, we have that $|V(H_1)|+|V(H_2)|+1=n$ and thus there exists $i \in \left\{ 1,2\right\}$  such that $|V(H_i)| \leq (n-1)/2+1$. Let $z$ be any node from $V(H_i)$ with $d_G(v,z) \geq 2$. If $z$ buys a link to $v$, the corresponding cost difference $\Delta C_{buy}$ associated to such deviation must satisfy the following inequality: 

$$ \Delta C_{buy} \leq \alpha - (d_G(v,z)-1)(n-|V(H_i)|) < \frac{n-1}{2}-(d_G(v,z)-1)\left( \frac{n-1}{2}\right) \leq \frac{n-1}{2}-\frac{n-1}{2} = 0$$

A contradiction with $G$ being a \NE graph. Therefore, $d_G(v,z) \leq 1$ for all $z \in V(H_i)$. Since $H_i \cup \left\{ v \right\}$ is biconnected it contains at least one cycle. Therefore, $H_i \cup \left\{ v \right\}$ must contain at least one edge $e=(w_1,w_2) \in E(H_i)$ such that $d_G(v,w_i)=1$ for $i=1,2$. If $w_1$ deletes $e$ then he gets one unit distance further from $w_2$ but he does not increment the distance units to any other node. Then the corresponding cost difference $\Delta C_{delete}$ satisfies the following inequality: 

$$\Delta C_{delete} \leq -\alpha + 1$$

If $\alpha \leq 1$ then $G$ is the star or the clique and the result is clear. Otherwise, then $\Delta C_{delete} < 0$ and therefore we reach a contradiction with $G$ being a \NE graph.

\end{proof}

\begin{proposition}\label{prop:biconnected}
For $\frac{n-1}{2} \leq \alpha < 4n$, if $H$ is any biconnected component of a non-tree \NE graph $G$ then: 
$$ diam(G) < diam(H)+ 994$$ 

\end{proposition}

\begin{proof}
In Proposition 5 from \cite{Alvarezetal3} it is shown that for $\alpha$ in the range $n < \alpha < 4n$ if $H$ is any biconnected component of a \NE graph $G$ then for any $w\in V(H)$ the maximum distance from $w$ to any node $w' \in S(w)$ is $125$. The only point in which the authors use $\alpha > n$ is in the reasoning for the subcase $(i)$ from which it is deduced that $|S(w)| \leq \frac{30}{31}n$. Thus it is easy to see that if we enlarge the range $n < \alpha < 4n$ to the range $\frac{n-1}{2} \leq \alpha < 4n$ then instead of $|S(w)| \leq \frac{30}{31}n$ we obtain $|S(w)| \leq \frac{61n+1}{62}n \leq \frac{123n}{124} $, where the last inequality comes from the fact that $n \geq 2$. The impact that it has to the rest of the result is that $125$ is changed to $ 497= 1+4 \frac{1}{1-\frac{123}{124}}$.  From here, the conclusion follows easily, since any maximal length path in $G$ is decomposed into two paths of length at most $497$ and a subpath of length at most $diam(H)$. \end{proof}

As a consequence of these observations we reach the following result:

\begin{theorem}\label{thm:biconnected}
There exists a constant $D_2$ such that every equilibrium graph $G$ of diameter larger than $D_2$ admits at most one biconnected component.
\end{theorem}

\begin{proof}

If $\alpha > 4n-13$ then we know that every \NE is a tree. Otherwise, if $\alpha \leq 4n-13$, then either $\alpha < \frac{n-1}{2}$ and $G$ admits at most one biconnected component due to Lemma \ref{lem:biconnected1} or if $\frac{n-1}{2} \leq \alpha < 4n$ then we shall see that the result works for some constant $D_2$. 

Indeed, let $H$ be a biconnected component of a \NE graph $G$ for $\alpha$ with $\frac{n-1}{2} \leq \alpha < 4n$. Looking at the proof of Proposition \ref{prop:biconnected} the maximum distance from $v\in V(H)$ to any node inside $S(v)$ is at most $497$. Therefore, the diameter of any subgraph contained in $S(v)$ with $v\in V(H)$ is at most $2\cdot 497=994$. In this way, if $G$ admits at least another biconnected component $H'$ different from $H$, then $H'$ must be contained inside some subset $S(v)$ with $v\in V(H)$. Hence, by the previous reasoning the diameter of $H'$ is at most $994$. Finally, invoking Proposition \ref{prop:biconnected}, $diam(G) < diam(H')+994 \leq 1998$. so that taking $D_2 = 1988$ the conclusion now is clear.

\end{proof}

\section{Conclusions}

These results regarding distance-uniformity and the number of non-bridge links bought by a single player, extend topological properties previously studied for a narrow interval of $\alpha$ to the whole remaining range of the parameter $\alpha$ for which it is not known whether the \PoA is constant. We knew that the diameter of equilibria is such a crucial quantity related to the \PoA. Now we know that the set of distances from a fixed node $u$ from any \NE graph follow a very restrictive and specific pattern, that does not depend on the choice of $u$. We knew that the average degree of equilibria is at most $2n/\alpha+4$. Now we know a little bit more of information, the number of non-bridge links bought by any node in a \NE graph of diameter large enough is at most $\max(R,2n/\alpha+6)$.

Finally, regarding the main techniques used to prove the strongest results, it has been useful to consider a massive set of deviations having  some sort of symmetry and then, upper and lower bound the sum of the corresponding cost differences. We think that this technique resembles the probability principle that has been used in \cite{Structurebasic}, for the sum basic network creation game. The fact that this technique works quite well in the \sumNCG, too, deserves further study in order to validate the constant \PoA conjecture if it is really true.

\end{document}